\documentclass[a4paper, USenglish]{lipics-v2018}
\usepackage[linesnumbered,noend]{algorithm2e}
\usepackage{graphicx}

\title{Double Threshold Digraphs}

\author{Peter Hamburger}{Department of Mathematics, Indiana-Purdue University\\{[Fort Wayne, IN 46805, USA]}}{hamburge@ipfw.edu}{}{}

\author{Ross M. McConnell}{Department of Computer Science, Colorado State University\\{[Fort Collins, CO 80523, USA]}}{rmm@cs.colostate.edu}{}{}

\author{Attila P\'or}{Department of Mathematics, Western Kentucky University\\{[Bowling Green, KY 42101]}}{attila.por@wku.edu}{}{}

\author{Jeremy P. Spinrad}{Department of Computer Science, Vanderbilt University\\{[Nashville, TN 37235, USA]}}{Jeremy.P.Spinrad@vanderbilt.edu}{}{}

\author{Zhisheng Xu}{Department of Computer Science, Colorado State University\\{[Fort Collins, CO 80523, USA]}}{xuzs9298@cs.colostate.edu}{}{}

\subjclass{Theory of computation $\rightarrow$ Mathematics of computing $\rightarrow$ Discrete Mathematics $\rightarrow$ Graph theory}


\authorrunning{P. Hamburger, R.~M. McConnell, A. P\'or, J.~P. Spinrad, Z. Xu}

\Copyright{Peter Hamburger and Ross M. McConnell and Attila P\'or and Jeremy P. Spinrad and Zhisheng Xu}

\newcommand{\RR}{{\mathbb{R}}}

\EventEditors{Igor Potapov, Paul Spirakis, and James Worrell}
\EventNoEds{3}
\EventLongTitle{43rd International Symposium on Mathematical Foundations of Computer Science (MFCS 2018)}
\EventShortTitle{MFCS 2018}
\EventAcronym{MFCS}
\EventYear{2018}
\EventDate{August 27--31, 2018}
\EventLocation{Liverpool, GB}
\EventLogo{}
\SeriesVolume{117}
\ArticleNo{69} 
\nolinenumbers 
\hideLIPIcs  

\begin{document}

\maketitle

\begin{abstract}\noindent
A {\em semiorder} is a model of preference relations
where each element $x$ is associated with
a utility value $\alpha(x)$, and there is a threshold $t$ such that 
$y$ is preferred to $x$ iff $\alpha(y) - \alpha(x) > t$.  
These are motivated by the
notion that there is some uncertainty in the utility values we
assign an object or that a subject may be unable to distinguish
a preference between objects whose values are close.
However, they fail to model the well-known phenomenon that preferences
are not always transitive.
Also, if we are uncertain of the utility values, it is not
logical that preference is determined absolutely by a comparison of them
with an exact threshold.  We propose a new model in which there are 
two thresholds, $t_1$ and $t_2$; if the difference 
$\alpha(y) - \alpha(x)$ is
less than $t_1$, then $y$ is not preferred to $x$; if the difference 
is greater than $t_2$ then $y$ is preferred to $x$; if it is between 
$t_1$ and $t_2$, then $y$ may or may not be preferred to $x$. 
We call such a relation 
a $(t_1,t_2)$ double-threshold semiorder, and the corresponding directed graph $G = (V,E)$
a $(t_1,t_2)$ double-threshold digraph.  Every directed acyclic graph is a 
double-threshold digraph; increasing bounds on $t_2/t_1$
give a nested hierarchy of subclasses of the directed acyclic graphs.
In this paper we characterize the subclasses in terms of forbidden
subgraphs, and give algorithms for finding an assignment 
of utility values that explains the relation in terms of a given $(t_1,t_2)$
or else produces a forbidden subgraph, and finding the minimum value 
$\lambda$  of $t_2/t_1$ that is satisfiable for a given directed 
acyclic graph.  We show that
$\lambda$ gives a useful measure of the complexity of a directed acyclic graph 
with respect to several optimization problems that are NP-hard on arbitrary 
directed acyclic graphs.
\end{abstract}

\section{Introduction}
A poset $P$ can be identified with a transitive digraph on its elements.
The poset $P = P(V,<)$ is a {\em semiorder}~\cite{Luce} if for some 
utility function $\alpha:V \rightarrow \RR$ we have $u <_P v$ if and only
if $\alpha(v) - \alpha(u) > 1$. Semiorders were introduced as a possible
mathematical model for preference in the social sciences.  A first possible 
model for preference is the {\em weak orders}, in which each element is 
assigned a utility value, such that $u$ is preferred to $v$ iff the value of 
$u$ is greater than the value of $v$.  This was viewed as too restrictive; 
many preference relationships cannot be modeled by a weak order. Semi-orders 
were designed to model imprecision in the valuation function; we may be 
indifferent between elements not only if they have exactly the same values, 
but also if the difference between the values is smaller than some threshold. 
There is a great deal of literature on the subject of semiorders and 
preference; see the 
books~\cite{Fis-85,Pirlot}.

Our original motivation for defining double-threshold digraphs comes from
an attempt to deal with an issue in mathematical psychology. Intuitively,
it is natural to think that preference is transitive; if one prefers $a$ to $b$
and $b$ to $c$, then one ``should'' prefer $a$ to $c$. However, a variety
of evidence exists showing that preferences are not always transitive.
This has led to a great deal of discussion; for a summary of this issue, see
\cite{Fis-91}. Viewpoints range
from the idea that the intuitive notion that preference is transitive
are simply wrong and must be thrown away entirely to questioning whether
what was being measured in the non-transitive findings was really a preference
relation. Between these two views, there has been work on finding mathematical
models that explain non-transitive preference; Fishburn \cite{Fis-91} gives
some possible models. 

One approach to mathematical modeling is to try to
give a reasonable model of extremely non-transitive preference; the famous
cyclic voter's paradoxes can be viewed as a model of preference which can
allow not just non-transitivity, but also cycles. 

Unlike these approaches, we generalize semi-orders to allow 
non-transitivity, but we require that the given set of preferences continue 
to be acyclic.  In other words, we consider any preference relation 
represented by a directed acyclic graph (a {\em dag}).
As in the case of semiorders, we assume that reported preferences are
influenced by an underlying hidden utility function, which may be approximate,
imperfectly known by a subject, or otherwise fail to capture all factors 
influencing a report of a preference.

One of our objectives is to 
obtain a measure of the departure of a given arbitrary acyclic set of pairwise
preferences from a model where preferences are driven exclusively by 
an underlying hidden utility function, as well as derive
an assignment of utility values that has the
most explanatory power, in a sense that we define within a new model
that we propose.

We propose a generalization of a semiorder, a {\em double-threshold semiorder}.
We loosen the definition of a semiorder to a broader class of relations
that are acyclic but not necessarily transitive, by 
allowing two thresholds $t_1$ and $t_2$ such that $t_1 \leq t_2$,
and finding a valuation
$\alpha(x)$ for each element $x$.  
For two elements $x$ and $y$,
$(x,y)$ is not reported as a preference if
$\alpha(y) - \alpha(x) < t_1$,
$(x,y)$ can freely be reported as a preference or not if
$t_1 \leq \alpha(y) - \alpha(x) \leq t_2$, and $(x,y)$ is reported
as a preference if $\alpha(y) - \alpha(x) > t_2$.
Let a {\em satisfying utility function} or a {\em satisfying assignment
of $\alpha$ values} for $(t_1, t_2)$ be a utility 
function $\alpha$ that meets these constraints.
This accommodates within the model the well-known phenomenon in the
literature on perception
that there can be a range of differences between the minimum
difference that is sometimes perceived and the minimum difference that is 
perceived reliably.

When the relation of the double-threshold semiorder is modeled by a dag, 
it is called a {\em double-threshold digraph}.
If a dag can be represented with thresholds $(t_1, t_2)$, then it can be
represented with any pair $(t'_1, t'_2)$ of thresholds such that
$t'_2/t'_1 = t_2/t_1$, since a solution $\alpha$ for $(t_1, t_2)$ can be turned
into a solution for $(t'_1, t'_2)$  by rescaling all $\alpha$ values 
by the factor $t'_1/t_1 = t'_2/t_2$.  Therefore,
for any pair $(t_1,t_2)$ of thresholds, the question of whether a particular
dag can be represented with them depends
on the ratio $r = t_2/t_1$; larger ratios allow representations of more
dags.

Henceforth, given a digraph $G$, let $n(G)$ denote the number of vertices 
and $m(G)$ the number of edges.  When $G$ is understood, we may denote 
these as $n$ and $m$.
For a dag $G$, let $\lambda(G)$ denote the minimum ratio of $t_2/t_1$ 
such that $G$ has a satisfying utility function for $(t_1,t_2)$.
When $G$ models a weak order, $t_1 = 1$ and $t_2 = \epsilon$ for
any $\epsilon > 0$ has a satisfying utility function.  For this
trivial special case, which is easily recognized in linear time,
we define $\lambda(G)$ to be 0, the lower bound on the satisfiable ratios
$t_2/t_1$, and call such a dag a {\em degenerate} dag.
All other dags are {\em nondegenerate}.

When $G$ or the preference relation it models is understood, we denote
$\lambda(G)$ simply by $\lambda$.  For a dag that models a nondegenerate 
semiorder, $\lambda = 1$; higher values of $\lambda$
provide a measure of the degree to which a given set of 
preferences depart from a semiorder.
An acyclic preference relation is a $(t_1, t_2)$-{\em semiorder} if
it has a satisfying utility function for $(t_1, t_2)$, that is, 
if $t_2/t_1 \geq \lambda$.  When such a preference relation is modeled as a 
digraph, we say the digraph is a $(t_1,t_2)$ {\em double-threshold digraph}.
We show that for any nondegenerate dag $G$, $\lambda(G)$ can be expressed as a ratio $j/i$
where $i$ and $j$ are integers such that $1 \leq i \leq j < i+j \leq n$
(Theorem~\ref{thm:ij}),
allowing $t_1$, $t_2$, and the utility function to have small integer
values.
Also, for any dag, $t_1 = 1$ and $t_2 = n-1$ is always satisfiable,
so $\lambda \leq n-1$.  An example of when the bound is tight
is when $G$ is a directed path.

Thus, the classes of dags with $\lambda$ bounded by different values 
give a nested hierarchy of dags, starting with weak orders and semiorders.
For each class  in the hierarchy, we
give a characterization of the class in terms of a set 
of forbidden subgraphs for the class.

When $G$ has no satisfying utility function for $t_1$, $t_2$, we
show how to return a forbidden subgraph as a 
certificate of this in $O(nm/r)$ time, where $r = t_2/t_1$,
and an $O(nm/\lambda)$ time bound for finding $\lambda$ (Theorem 18).
The algorithm combines elements of the Bellman-Ford single-source shortest
paths algorithm~\cite{Cormen}, Karp's minimum mean cycle algorithm~\cite{karp78},
and dynamic programming techniques based on a topological sort of a dag.
For $t_2/t_1 = \lambda$, a satisfying assignment, together with a forbidden
subgraph for a smaller ratio, give a 
certificate that $\lambda = t_2/t_1$, and these take $O(nm/\lambda)$ time 
to produce.

If $\lambda$ is less than $2$, $G$ must be transitive.
The converse is not true:
it is easy to show that the class of posets does not have bounded
$\lambda$.
Consider a chain $(v_1, v_2, \ldots, v_{n-1})$ in a poset and a 
vertex $v_n$ that is incomparable to the others;
$t_2 \geq t_1(n-2)/2$.  
Even though they are transitive, some posets are not good models of
a preference relation that is based on an underlying utility function.

Although we show that bounding $\lambda$ can make
some NP-complete problems tractable, bounded-ratio double-threshold digraphs
are in one sense enormously larger than semiorders.  Semiorders
correspond to digraphs that can be represented with ratio 1.
These classes of digraphs
both have implicit representations \cite{spinrad}, implying that there are
$2^{O(n \log n)}$ such digraphs on a set of $n$ labeled vertices. 
By contrast, every height 1
digraph can be represented with ratio $1$:
for each vertex $x$, assign $\alpha(x) = 0$ if is it a source
or $\alpha(x) = 1$ if it is a sink
and make the thresholds $t_1 = t_2 = 1$.
The number of such digraphs on $n$ labeled vertices, hence the number with 
ratio $\lambda$ for any $\lambda$ greater than 
or equal to 1, is $2^{\Theta(n^2)}$.  

The {\em underlying undirected graph} of a dag is the symmetric
closure, that is, the undirected
graph obtained by ignoring the orientations of the edges.
In this paper, we say that a dag is {\em connected} if its underlying
undirected graph is connected.  Similarly,
by a {\em clique}, {\em coloring}, {\em independent set}, or
{\em clique cover} of a dag, we mean a clique, coloring, independent
set or clique cover of the underlying undirected graph.
Hardness results about these problems on undirected graphs
also apply to dags, since every undirected graph $G$ is the underlying 
undirected graph of the dag obtained by assigning an acyclic orientation
to $G$'s edges.

Finding a maximum independent set or clique in a dag
takes polynomial time if the dag is transitive (a poset), hence if it is
a semiorder, but for arbitrary dags, there is no polynomial-time 
approximation algorithm for finding 
a independent set or clique whose size is within a factor of 
$n^{1-\epsilon}$ of the largest
unless P = NP~\cite{hastad99}.
However, for a connected dag $G$, we give an
$O(\lambda m^{\lfloor \lambda + 1 \rfloor/2})$ algorithm for finding
a maximum clique (Corollary~\ref{cor:exact}), and an approximation algorithm that finds a clique whose
size is within a desired factor of $i$ of that of a maximum clique
in $O(nm/\lambda + m^{\lfloor \lambda/i + 1 \rfloor/2})$ time (Corollary~\ref{cor:cliqueApprox}).

We show that finding a maximum independent set is still NP-hard when 
$\lambda \geq 2$, but we give a polynomial-time approximation algorithm 
that produces an independent set whose size is within a factor of 
$\lfloor \lambda +1 \rfloor$ of the optimum (Theorem~\ref{thm:IS}).
We give approximation bounds of $\lfloor \lambda+1 \rfloor$ for minimum 
coloring and minimum clique cover (Theorems~\ref{thm:coloring} and~\ref{thm:cCover}), which also have no polynomial algorithms 
for finding an $n^{1 - \epsilon}$ approximation for arbitrary dags unless 
P = NP.  

Thus, restricting attention to dags such that $\lambda$
is bounded by a constant makes some otherwise NP-hard problems easy
and gives rise to polynomial-time approximation algorithms that cannot exist
in general unless P = NP.  In each case, the time bound or the approximation
bound is an increasing function of $\lambda$.
This supports the view of $\lambda$ as a measure of complexity of a dag.
By contrast, for most similar attempts to measure complexity of a
graph or digraph, the measurement is NP-hard to compute;
examples include dimension of a poset, interval number, boxicity, and
many others; see \cite{spinrad}.

A concept similar to $\lambda$ was given previously by Gimbel and Trenk 
in~\cite{GimbelTrenk98}.  They developed a generalization of
weak orders to partial orders that corresponds to the special case
of a $(1,k)$ transitive dag.  Not assuming transitivity requires us
to use different algorithmic methods, but our bounds improve their 
bounds for their special case from $O(n^4k)$ and $O(n^6)$ to $O(mn/k)$.
Most of their structural results are disjoint from ours because they
are relevant to partial orders and their underlying undirected graphs,
the {\em comparability graphs.}

\section{Satisfying utility functions and forbidden subgraphs}\label{sect:polyUtility}

We give the following formal definition:

\begin{definition}
	A dag is a $(t_1,t_2)$ double-threshold digraph if there exists an
	assignment of a real value $\alpha(v)$ to each vertex $v$ such that
	whenever $(u,v)$ is an edge, $\alpha(v) - \alpha(u) \geq t_1$
	and whenever $(u,v)$ is not an edge, $\alpha(v) - \alpha(u) \leq t_2$.
\end{definition}

Whether the constraints can be satisfied can be formulated as the problem
of finding a feasible solution to a linear program:

\begin{itemize}
	\item  $\alpha(v) - \alpha(u) \geq t_1$ for each $(u,v)$ such that $(u,v)$ is an edge;

	\item  $\alpha(v) - \alpha(u) \leq t_2$ for each $(u,v)$ such that neither $(u,v)$ nor $(v,u)$ is an edge;

\item $\alpha(v) \leq 0$ for all $v \in V(G)$.
\end{itemize}

The last constraint is added as a convenience;  for any satisfying
assignment, an arbitrary constant can be subtracted from all of
the $\alpha$ values to obtain a new satisfying assignment, so the constraint
cannot affect the existence of a feasible solution.

This is a special case of a linear program, 
a {\em system of difference constraints}, where
each constraint is an upper bound on the difference of two
variables.  
This reduces to the problem of finding the weight of a least-weight
path ending at each vertex in a digraph derived from the constraints, as 
described in~\cite{Cormen}, where there is a satisfying assignment
if and only if the digraph of the reduction has no negative-weight cycle.   
Applying the reduction to the problem of determining whether there is
a satisfying utility function on $G$
yields a digraph $G_d$, where $V(G_d) = V(G)$
(see Figure~\ref{fig:Gd}).  
$G_d$ has an edge $(y,x)$ of weight $-t_1$ for each edge $(x,y)$
of $G$, and edges $(u,v)$ and $(v,u)$ of weight $t_2$ for each pair
$\{u,v\}$ such that neither of $(u,v)$ and $(v,u)$ is an edge of $G$.
A negative cycle in $G_d$ proves that the system is not satisfiable; 
otherwise, for each $x \in V$, assigning $\alpha(x)$ to be the minimum 
weight of any path ending at $x$ gives a satisfying assignment for 
$(t_1, t_2)$.  

The single-source least-weight paths problem where some weights are negative
can be solved in $O(nm)$ time, but $G_d$ has $\Theta(n^2)$ edges,
so a direct application of this approach takes $\Theta(n^3)$ time to find 
a satisfying assignment or produce a negative-weight cycle in $G_d$.  We derive
tighter bounds below.

\begin{figure}
\centerline{\includegraphics[]{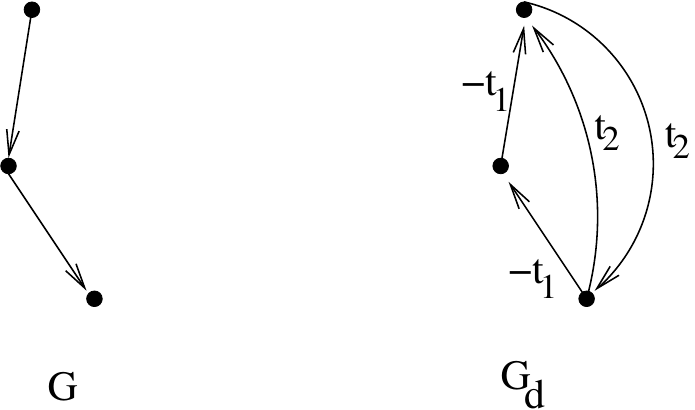}}

\caption{Reduction of finding a satisfying utility function to
the single-source least-weight paths problem.  Edges of 
weight $t_1$ in $G_d$ are acyclic.}\label{fig:Gd}
\end{figure}

In terms of $G$, a negative cycle of $G_d$ translates to a forbidden
subgraph characterization of $(t_1, t_2)$ double-threshold digraphs:

\begin{definition}
	Let $(u,v)$ be a {\em hop} in $G$ if neither $(u,v)$ nor $(v,u)$
	is an edge of $G$.
	Let a {\em forcing cycle} be a simple cycle $(v_1, v_2, ..., v_k)$
	such that such that for each consecutive pair $(v_i, v_{i+1})$ (indices mod $k$), 
        the pair is either a directed edge of $G$ or a hop.
	Let the {\em ratio} of the forcing cycle be the ratio of the number of edges to
	the number of hops.
\end{definition}

\begin{theorem}
	For a nondegenerate dag $G$, the minimum satisfiable ratio $\lambda$ 
	is equal to the maximum ratio of a forcing cycle in $G$.
\end{theorem}

One consequence of the theorem is that when $G$ is a nondegenerate dag, a satisfying
assignment of $\alpha$ values for thresholds $(t_1,t_2)$, together
with a forcing cycle with ratio equal to $t_2/t_1$ gives
a certificate that $\lambda(G) = t_2/t_1$, as illustrated
in Figure~\ref{fig:forbiddenCycle}.

\begin{figure}
\centerline{\includegraphics[]{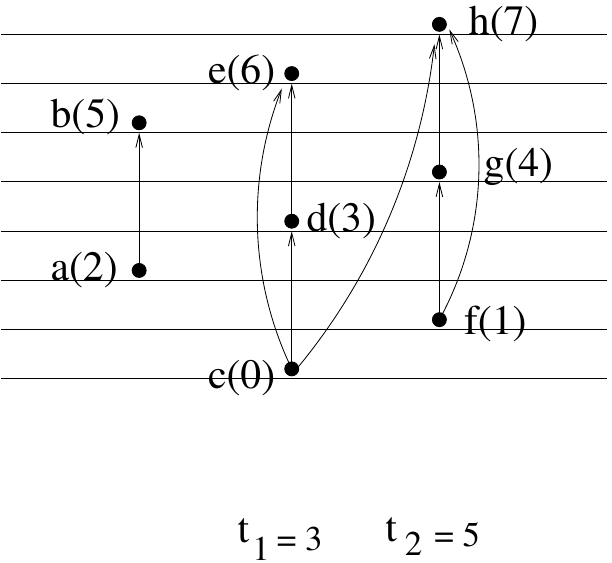}}
\caption{A dag such that $\lambda = 5/3$.  The
number next to each vertex is the value of the utility function, conforming 
to $t_1 = 3$ and $t_2 = 5$.  The cycle $(a,b,c,d,e,f,g,h)$ is a cycle
of directed edges and hops in which the ratio of edges to hops
is 5/3.  Since $\lambda < 2$, the dag is transitive.
}\label{fig:forbiddenCycle}
\end{figure}

\begin{theorem}\label{thm:ij}
	For every nondegenerate dag $G$, $\lambda(G)$ can be expressed as a ratio $i/j$
	of integers such that $1 \leq i \leq j < i+j \leq n$.
\end{theorem}
\begin{proof}
This follows from the fact that $\lambda \geq 1$ and is the ratio of the 
number $j$ of edges to the number $i$ of hops on a forcing cycle.
\end{proof}

Aside from showing that optimum values of $t_1$ and $t_2$ can be
expressed as small integers, the theorem
gives an immediate $O(n^3 \log n)$ bound
for finding $\lambda$.  This is because it implies that the number of possible
values $j/i$ that $\lambda$ can take on is $O(n^2)$, and that these can be
generated and sorted in $O(n^2 \log n)$ time.  A binary search on this
list, spending $O(n^3)$ time at each probe $j/i$ to determine whether
$G$ is an $(i,j)$ double-threshold digraph, as described above, can then be used
to find $\lambda$.
Once $\lambda$ is known, a satisfying assignment of utility values
for $t_2/t_1 = \lambda$, together with a forcing cycle with forcing
ratio equal to $\lambda$ gives a certificate that the claimed value of
$\lambda$ is correct.
We improve these bounds to $O(nm/\lambda)$ in section~\ref{sect:faster}.

\section{$k$-clique extendable orderings}

In the book~\cite{spinrad}, Spinrad introduced the class of 
{\em $k$-clique extendable orderings}
of the vertices of graphs, which we explain below.  Finding whether a graph
has a 2-clique extendable ordering takes polynomial time, but no polynomial time
bounds are known for $k \geq 3$.
However, we show in the next section that a topological
sort of a nondegenerate dag $G$ is a $k$-clique extendable ordering for 
$k = \lfloor \lambda(G) \rfloor + 1$, and
develop several applications of this result to optimization problems.
In this section, we give the details and analysis of the time
bound of an algorithm suggested in~\cite{spinrad} 
for finding a maximum clique, given a $k$-clique extendable ordering.

Two sets {\em overlap} if they intersect and neither is a subset of the other.
Let $\sigma = (v_1, v_2,$ $\ldots,$ $v_n)$ be an ordering of the vertices of a 
graph, $G = (V,E)$.  For $U \subseteq W \subseteq V$ let us say
that $W$ {\em ends with $U$} if the elements of $U$ are the last elements
of $W$ in $\sigma$, that is, if no element of $W \setminus U$ occurs
after an element of $U$.  $W$ {\em begins with $U$} 
if $W$ ends with $U$ in $(v_n, v_{n-1}, \ldots, v_1)$.

\begin{definition}
An ordering $\sigma = (v_1, v_2, \ldots, v_n)$ of vertices of a graph 
$G = (V,E)$ is $k$-clique extendable ordering of $G$ if, whenever $X$ and $Y$
are two overlapping cliques of size $k$, $|X \cap Y| = k-1$, and  $X \cup Y$ 
begins with $X \setminus Y = \{a\}$ and ends with $Y \setminus X = \{b\}$ 
in $\sigma$, then $a$ and $b$ are adjacent and $X \cup Y$ is a clique.  
\end{definition}

This is a generalization of transitivity, since a
dag is transitive if and only if its topological sorts are two-clique
extendable orderings, hence a graph is a comparability graph if and only if
it has a two-clique extendable orderings.
In~\cite{spinrad}, it is shown that three-clique extendable orderings arise
naturally in connection with visibility graphs, and that it takes
polynomial time to find a maximum clique in a graph, given a three-clique
extendable ordering.  A polynomial-time generalization for $k$-clique 
extendable orderings is implied; we give details and a time bound next.

\begin{lemma}\label{lem:kext}
If $\sigma = (v_1, v_2, \ldots, v_n)$ is a $k$-clique extendable ordering of
a graph $G$ and $X$ and $Y$ are overlapping cliques of any size greater than or
equal to $k$, such that $|X \cap Y| \geq k-1$ and $X \cup Y$ begins with 
$X \setminus Y$ and ends with $Y \setminus X$ in $\sigma$, then $X \cup Y$ is 
a clique.
\end{lemma}

\begin{proof}
It suffices to show that every element of $X \setminus Y$ is adjacent
to every element of $Y \setminus X$.  Let $x$ be an arbitrary element
of $ X \setminus Y$, $y$ be an arbitrary element of
$Y \setminus X$, and $Z$ be any $k-1$ elements of $X \cap Y$.
Then $\{x\} \cup Z$ and $Z \cup \{y\}$ are two $k$-cliques and, 
by the definition of a $k$-clique extendable ordering,
their union is a clique, and $x$ and $y$ are adjacent.
\end{proof}

\begin{corollary}\label{cor:DP}
If $\sigma = (v_1, v_2, \ldots, v_n)$ is a $k$-clique extendable ordering of
a graph $G$, $X$ is a $k$-clique ending with $\{v\}$ and $Z$ is a largest
clique of $G$ ending with the $(k-1)$-clique $X \setminus \{v\}$, then 
$Z \cup \{v\}$ is a largest clique of $G$ ending with $X$.
\end{corollary}
\begin{proof}
For any clique $Y$ ending with $X$, $Y \setminus \{v\}$
is a clique ending with $X \setminus \{v\}$.
$Z \cup \{v\} = Z \cup X$, which is a clique by Lemma~\ref{lem:kext}.
\end{proof}

Corollary~\ref{cor:DP} is the basis of the recurrence for a dynamic programming 
algorithm for finding a maximum clique of $G$, given a $k$-clique extendable
ordering.  We enumerate all
$k$-cliques and then label each $k$-clique $K$ with the maximum size 
$h_K$ of a clique that ends with $K$.  If $(u_1, u_3, \ldots, u_k)$ is the
left-to-right ordering of a $k$-clique in the ordering, then its label 
is one plus the maximum of the labels of cliques of the form 
$(x, u_1, u_2, \ldots, u_{k-1})$.  The size of the maximum clique of $G$ is
the maximum of the labels.  Details and the proof of the following resulting 
time bound is given in the appendix.

\begin{theorem}\label{thm:kextTime}
Given a $k$-clique extendable ordering of a graph $G$, a maximum clique
can be found in $O(km^{k/2})$ time.
\end{theorem}

It is easy to see that when the vertices of $G$ have positive weights, the 
problem of finding a maximum weighted clique can be solved in the same time bound,
using a trivial variant of Corollary~\ref{cor:DP}.

\section{Optimization problems on dags with bounded $\lambda$ values}\label{sect:opt}

We now show that restricting attention to dags such that $\lambda$
is bounded by a constant makes some otherwise NP-hard problems easy
or gives rise to polynomial-time approximation algorithms that cannot exist
for the class of all dags unless P = NP.  The NP-hard problems we consider
can be trivially solved in linear time on degenerate dags, so we focus
on nondegenerate dags.

\begin{theorem}\label{thm:k-extendable}
Let $G$ be a nondegenerate dag and $k = \lfloor \lambda(G) \rfloor + 1$.
A topological sort of $G$ is
a $k$-clique extendable ordering.
\end{theorem}
\begin{proof}
Let $(v_1, v_2, \ldots, v_n)$ be a topological sort, and
let $\alpha$ be a satisfying utility function for $(t_1,t_2)$
such that $t_2/t_1 = \lambda$.
Let $(w_1,w_2, \ldots, w_k)$ and $(w_2, w_3, \ldots, w_k, w_{k+1})$ be the 
left-to-right orderings of two $k$-cliques $K'$ and $K$.
Then $(w_1, w_2, \ldots, w_{k+1})$ is a directed path in $G$, hence
$\alpha(w_{k+1}) - \alpha(w_1) \geq kt_1 > t_2$, $(w_1, w_{k+1})$
is an edge and $K \cup K'$ is a clique.
\end{proof}

\begin{corollary}\label{cor:exact}
It takes $O(\lambda m^{\lfloor \lambda+1 \rfloor / 2})$
time to find a maximum clique in a connected nondegenerate dag $G$.
\end{corollary}
\begin{proof}
To avoid an additive $O(nm/\lambda)$ term, run
the dynamic programming algorithm on a topological sort under the assumption 
that it is a 2-clique extendable ordering in $O(m)$ time by 
Theorem~\ref{thm:kextTime}, and return the result if 
it is a clique.
Otherwise, do the same under the assumption that it is a 3-clique extendable ordering,
in $O(m^{3/2})$ time.  If a max clique
has not yet been returned, then $\lambda \geq 3$ by Theorem~\ref{thm:k-extendable},
so compute $\lambda$ in $O(nm/\lambda) = O(m^2)$ time, which is now subsumed
by the bound we want to show.
A topological sort is a $\lfloor \lambda \rfloor +1$ extendable ordering
by Theorem~\ref{thm:k-extendable}, so 
it takes $O(\lambda m^{\lfloor \lambda+1 \rfloor / 2})$ time to find a maximum clique 
by Theorem~\ref{thm:kextTime}.
\end{proof}

Even if $\lambda$ is bounded by a moderately large constant,
this bound could be prohibitive in 
practice, but it also gives an approximation
algorithm that allows a tradeoff between time and approximation factor:

\begin{corollary}\label{cor:cliqueApprox}
Given a connected nondegenerate dag $G$ and integer $i$ such that $1 \leq i \leq \lambda$, 
a clique whose size is within 
a factor of $i$ of the size of a maximum clique can be found
in $O((\lambda/i) m^{(\lfloor \lambda/i \rfloor + 1 )/2})$ time.
\end{corollary}
\begin{proof}
Let $G'$ be the result of removing the edges $\{(u,v)| (u,v) \in E(G)$
and $\alpha(v) - \alpha(u) < i\}$.  A satisfying function $\alpha$ for
$G$ and thresholds $(1, \lambda(G))$ is also a satisfying function for $G'$ and 
thresholds $(i, \lambda(G))$, so $\lambda(G') \leq \lambda(G)/i$.  Applying 
Theorems~\ref{thm:kextTime} 
and~\ref{thm:k-extendable}, we get a maximum clique of $G'$
in $O((\lambda/i) m^{(\lfloor \lambda/i \rfloor + 1 )/2})$ time.
A maximum clique of $G$ induces a directed path $(v_0, v_1, ..., v_k)$ in $G$, 
and $\{v_0, v_i, v_{2i}, \ldots, v_{\lfloor k/i \rfloor}\}$ is a clique of $G'$,
so the size of a maximum clique in $G'$ is within a factor of $i$ of the 
size of a maximum clique in $G$.
\end{proof}

If $\lambda(G) < 2$,
a maximum independent set in $G$ can be obtained in polynomial time, 
since $G$ is transitive~\cite{gavril00}.  
However, even when $\lambda(G) = 2$,
the problem of determining whether $G$ has an independent set of size $k$
is NP-complete.  This is seen as follows.  It is NP-complete
to decide whether a 3-colorable graph has an independent set of
a given size $k$, even when the 3-coloring is given~\cite{kratochvilNesetril90}.
Given such a graph $G'$, $k$, and three-coloring, let $C_1$, $C_2$, and $C_3$ 
be the three color classes.
Every edge $e$ has endpoints in two of the classes; orient $e$
from the endpoint in the class with the smaller subscript to the
endpoint in the class with the larger subscript.
Doing this for all edges results in a dag $G$ such that
$\lambda(G) = 2$, since, for each vertex $x$, if $x \in C_i$,
assigning $\alpha(x) = i$ gives a satisfying assignment of utility
values for $t_2 = 2$ and $t_1 = 1$.
There is an independent set of size $k$ in $G$ if and only if there
is one in $G'$.

\begin{theorem}\label{thm:IS}
For $G$ in the class of dags where 
$\lfloor \lambda(G) \rfloor + 1  \leq k$, there is a polynomial
$k$-approximation algorithm for the problem of finding a maximum
independent set in $G$.
\end{theorem}
\begin{proof}
Find a satisfying assignment of utility values
for $(t_1, t_2)$ such that $t_2/t_1 = \lambda(G)$, 
then find an interval of the form $[x, x+ t_1)$ such that the size of the set
$Y$ whose $\alpha$ values are in the interval is maximized.
$Y$ is an independent set, since no pair of them has $\alpha$ values that 
differ by $t_1$.  Return these vertices as an independent set.

For the approximation bound, let $X$ be a maximum independent set.  
The $\alpha$ values of $X$ lie in an interval of the form $[y, y + t_2]$,
which is a subset of the union
$[y, y + kt_1)$, of $k$ intervals of the form $[x, x+t_1)$, hence
$|X| \leq k|Y|$.
\end{proof}

Proofs of the following make similar use of the availability of
satisfying $\alpha$ values are given in the appendix.

\begin{theorem}\label{thm:coloring}
For $G$ in the class of dags where $\lfloor \lambda(G) \rfloor + 1 \leq k$,
there is a polynomial
$k$-approximation algorithm for the problem of finding a minimum
coloring of $G$.
\end{theorem}

\begin{theorem}\label{thm:cCover}
For $G$ in the class of dags where 
$\lfloor \lambda(G) \rfloor +1 = k$, there is a polynomial
$k$-approximation algorithm for the problem of finding a minimum
clique cover of $G$.
\end{theorem}

\section{$O(nm/\lambda)$ bounds for finding satisfying utility 
functions, $\lambda$, and certificates}\label{sect:faster}

In this section, we first show how to find a satisfying assignment of utility 
values for given thresholds $(t_1,t_2)$, in $O(nm/r)$ time, where
$r = t_2/t_1$.  We then show how to find $\lambda$ in $O(nm/\lambda)$
time.  By solving the second problem to find $\lambda$, then
selecting $(t_1,t_2)$ such that $t_2/t_1 = \lambda$ and solving the first,
we get the certificates for $\lambda$, that is, a satisfying assignment
and a cycle such that the ratio of edges to hops is $\lambda$, which
comes from a zero-weight cycle in $G_d$.

For both of these problems, we use the following.
When $G$ is an arbitrary digraph where each vertex $x$ has a weight
$w(x)$ and each edge $(y,z)$ has a weight $w(y,z)$, it takes
$O(m)$ time to find $w'(v) = \min(\{\infty\} \cup \{w(u) + w(u,v)|(u,v)$
is an edge of $G\}$ for each vertex $v$ of $G$.
Let us call this the {\em general relaxation procedure}.
In the special case where $G$ is a dag, it takes $O(m)$ time to find
$w'(v) = \min_u (\{w(u) + w_{uv})\})$,
where $w_{uv}$ is the minimum weight of any path from $u$ to $v$ and
$w_{vv}$ = 0.
This can be used to solve the single-source shortest paths problem
on a connected dag in $O(m)$ time~\cite{Cormen}.
Let us call this the {\em dag variant} of the relaxation procedure.

In a digraph with edge weights, let the {\em length} of a walk be the
number of occurrences of edges on the walk and its {\em weight} be
the sum of weights of occurrences of edges.  If an edge occurs $k$ times
on the walk, it contributes $k$ to the length, and if its weight
is $w$, it contributes $kw$ to the weight and $kw$ to the
number of (occurrences of) edges of weight $w$ on the walk.

\subsection{Finding a satisfying utility function or a forbidden subgraph for
$(t_1,t_2)$.}\label{sect:BF}

The Bellman-Ford algorithm is a dynamic
programming algorithm that works as 
follows on a connected digraph $G$ where a vertex $s$ has been added
that has an edge of weight zero to all other vertices.
Let $D(i,v)$ be the minimum weight of any walk from $s$ to $v$
that has at most $i+1$ edges.  $D(i,v)$ is just the minimum weight of any
walk of length at most $i$ in $G$ ending at $v$; henceforth we omit $s$
from the discussion.  $D(0,v) = 0$ for all $v \in V$.
During the ``$i^{th}$ pass'' the algorithm computes $D(i,v)$ 
as $min(\{(D(i-1,u) + w(u,v)| (u,v) \in E\})$.  This is just an
instance of the general relaxation procedure where $w(v) = D(i-1,v)$
and the loop $(v,v)$ is considered to be an edge of weight 0 for each $v \in V$.
If there is no negative cycle, there is always
a path ending at $v$ that is a minimum-weight walk ending at $v$, so 
$D(n-1,v)$ gives the minimum weight
of any path ending at $v$.  If there is a negative cycle, this
is detected when $D(n,v) < D(n-1,v)$ for some $v$, indicating a walk
of length $n$ of smaller weight of any path, which must have a negative
cycle on it.  By annotating the dynamic programming entries with suitable
pointers, it is possible to find such a cycle within the same bound.
The $n$ passes to compute $D(n,v)$ for all $v$ each
take $O(m)$ time, for a total of $O(nm)$ time.

To exploit the structure of $G_d$ to improve the running time,
we let $B(i,v)$ denote the minimum weight of any path that has at most
$i$ edges of weight $t_2$, rather than at most $i$ edges in total.
We use the elements $B(i,v)$, rather than the elements of $D(i,v)$,
as the elements of the dynamic programming table.
Let us call this {\em reindexing the dynamic
programming table}.  We obtain $B(0,v)$ by assigning $w(v) = 0$
and running the dag variant of the relaxation procedure on the edges of
weight $-t_1$, since they are acyclic.
For pass $i$ such that $i > 0$, any improvements obtained 
by allowing an $i^{th}$ edge of weight $t_2$ are computed with
the general relaxation procedure, where loops are
considered to be edges of weight 0, and, after this, 
any additional improvements obtained by appending additional edges
of weight $-t_1$ are computed by the dag
variant of the relaxation procedure.

Because every vertex has a walk of length and weight 0 ending at it,
$B(i,v) \leq 0$ for $i \geq 0$.  Therefore, for $i > 0$, if $B(i,v) < B(i-1,v)$,
the ratio of edges of weight $-t_1$ to edges
of weight $t_2$ is greater than $r = t_2/t_1$.  
Any such walk
must have more than $ir$ edges of weight $-t_1$, hence length greater than
$i (r+1)$.
Therefore, if there is no negative cycle in $G_d$, for
$i = \lfloor (n-1)/(r+1) \rfloor + 1$, $B(i,v) = D(n-1,v)$, and
a negative cycle occurs if $B(i+1,v) < B(i,v)$ for this $i$ and
some $v$.  A negative cycle can be found by the standard technique of annotating
the results of the relaxation operations with pointers to earlier
results.  The advantage of reindexing the table is that the 
algorithm now takes $O(n/r)$ passes instead of $n$ of them.   

To get the $O(nm/r)$ bound, it remains to show how to perform each 
pass in $O(m)$ time.  The bottleneck is evaluating
$w'(v) = \min \{\{w(v)\} \cup \{w(u) + t_2| (u,v)$ is an edge
of weight $t_2\}$ for the general relaxation step.
Since all of the edges have the same weight, we
rewrite this as $w'(v) = min \{w(v), w(x) + t_2\}$
where $x$ minimizes $w(u) = B(i-1,u)$ for all $u$ such that
$w(u,v) = t_2$.
To evaluate this, we just have to find $x$.
At the beginning of the pass, we radix sort the vertices
in ascending order of $B(i-1,*)$, giving
list $L$.  To compute $B'(i,v)$, we mark the vertices that have
an edge to $v$, then traverse $L$ until we find $x$ as the first
unmarked vertex we encounter, then unmark the vertices that have
edges to $v$.  
This takes time proportional to the in-degree of $v$,
hence $O(m)$ time for all vertices in the pass.

\subsection{Finding $\lambda$}

To find $\lambda(G)$, we use the fact that that if $t_2/t_1 = \lambda$,
the corresponding weighting of $G_d$ will give it a zero-weight cycle in
$G_d$, which gives a forcing cycle of ratio $\lambda$ in $G$ as a certificate.

For arbitrary $(t_1, t_2)$, let
the {\em mean weight} of a directed cycle or path of length at least one
in $G_d$ be the weight of the cycle divided by the number of edges.  The 
minimum mean weight of a cycle is the {\em minimum cycle mean}.  Subtracting 
a constant $c$ from the weight of all edges in $G_d$ subtracts $c$ from the 
mean weight of every cycle and path of length at least one.
For arbitrary $t_1$ and $t_2$, weighting $G_d$ in accordance with 
$(t_1 + c,t_2 - c)$ in place of $(t_1,t_2)$ has the same effect of 
subtracting $c$ from the weights of 
all edges.  Thus, for arbitrary $(t_1,t_2)$, if $c$ is the minimum 
cycle mean of the corresponding weighting of $G_d$, then 
$\lambda = (t_2 - c)/(t_1 + c)$.  Finding $\lambda$ reduces to
finding the minimum cycle mean in the weighting of $G_d$ obtained from
an arbitrarily assigned $(t_1, t_2)$.

In a digraph $G$ with edge weights,
let $F(i,v)$ be the minimum weight of any 
walk of length {\em exactly} $i$ ending at $v$.
In~\cite{karp78}, Karp showed the following:

\begin{theorem}\label{thm:karp}
The minimum cycle mean of a digraph with edge weights is \\
$\min_{v \in V} \max_{0 \leq i < n} [(F(n,v) - F(i,v))/(n-i)]$.
\end{theorem}

Karp actually shows this when an arbitrary vertex $s$ is selected
and $F(i,v)$ is defined to be the minimum weight of all walks of
length $i$ from $s$ to $v$, but if it is true for walks beginning
at an arbitrary vertex $s$,
then it is true when $s$ is allowed to vary over all vertices of $V$.
Omitting $s$ from consideration in this way in his proof gives a direct proof 
of this variant of his theorem.
He reduces the problem to the special case where $G$ is strongly
connected by working on each strongly-connected component separately,
but the only purpose of this in his proof is to ensure that there is a path 
from $s$ to all other vertices, and this is unnecessary when $s$ is allowed 
to vary over all vertices.  

$F(i,v)$ can be computed by a variant of Bellman-Ford, by 
using the recurrence
$F(i,v) = \min(\{\infty\} \cup \{F(i-1,u) + w(u,v)| (u,v) \in E\})$
in place of $D(i,v) = \min(\{D(i-1,v)\} \cup 
\{D(i-1,u) + w(u,v)| (u,v) \in E\})$.  The only difference from the
algorithm of Section~\ref{sect:BF} is that 
loops of the form $(v,v)$ are not considered to be edges.  Computing 
$F(n,v)$ for all $v \in V$ takes $n$ passes, each of which applies the 
general relaxation operation, for a total of $O(nm)$ time.

An obstacle to an $O(nm/\lambda)$ bound that we did not have in
Section~\ref{sect:BF} is that 
in Theorem~\ref{thm:karp}, computing
$[(F(n,v) - F(i,v))/(n-i)]$ for $0 \leq i < n$
requires $\Theta(n^2)$ computations,
which is not $O(nm/\lambda)$.

We again reindex the dynamic programming
table (Section~\ref{sect:BF}), letting $H(i,v)$ denote the minimum-weight
walk ending at $v$ in $G_d$ that has {\em exactly} $i$ edges of weight $t_2$.
We compute the values in passes, computing $H(i,v)$ for each $v \in V$ 
during pass $i$.  As in Section~\ref{sect:BF}, each pass takes
takes $O(m)$ time; the only change is that in the general relaxation
step, loops are not considered to be edges.
We claim that $O(n/\lambda)$ passes suffice, but a new difficulty
is knowing when to stop, since, unlike $r$ of the Section~\ref{sect:BF},
$\lambda$ is not known in advance.

A walk with $i$ edges of weight $t_2$ and weight $H(i,v)$ has $i$ edges of 
weight $t_2$, so it must have $(it_2 - F(i,v))/t_1$ edges of weight $-t_1$.
Its length, $l(i,v)$, can be computed as $i + it_2 - F(i,v)$ in $O(1)$ time. 

Let a term $H(i,v)$ be {\em term of interest} if $l(i,v) = n$, that is, if
it corresponds to a {\em walk of interest} of length $n$.
We use the following reindexed variant of Karp's theorem, which 
says that it suffices to compute an inner maximum over a smaller set,
and only for terms of interest.  The proof is the one Karp gives, 
reindexed, and omitting reference to a start vertex $s$ by allowing
the start vertex to vary over all vertices.  For completeness, we
give the modified proof in the appendix.

\begin{theorem}\label{thm:karp2}
In $G_d$, the minimum mean weight of a cycle is equal to \\
$\min_{\{(i,v)| l(i,v) = n\}} \max_{0 \leq j < i} (H(i,v) - H(j,v))/(n  - l(j,v))$
\end{theorem}

The solution is given as Algorithm~\ref{alg:MCM}.
During the $i^{th}$ pass, the algorithm computes 
$H(i,v)$ for all $v \in V$.
Before proceeding to the next pass, it updates a partial
computation of the expression of Theorem~\ref{thm:karp2}, computing
$\max_{0 \leq j < i} (H(i,v) - H(j,v))/(l(i,v) - l(j,v))$ for
each the terms of interest $H(i,v)$ that has been computed during the pass,
and keeping track of the minimum of these computations so far.
Let a term of interest $H(i,v)$ be {\em critical} if 
the minimum cycle mean is equal to 
$\max_{0 \leq j < i} (H(i,v) - H(j,v))/(l(i,v) - l(j,v))$.
The strategy of the algorithm is to return
the minimum it has found so far once it
detects that a critical term has been evaluated.

Let a {\em critical walk} be a walk of length $n$ giving
rise to a critical term.  

\begin{lemma}\label{lem:meanBound}
In $G_d$, the mean weight of a critical walk is less than or equal to the 
minimum cycle mean.
\end{lemma}

The proof is given in the appendix.

\begin{theorem}
	Given a nondegenerate dag $G$, it takes $O(nm/\lambda)$ time to find $\lambda(G)$.
\end{theorem}\label{thm:findLambda}

\begin{proof}
The basis of this is Algorithm~\ref{alg:MCM}.
For a term of interest,
$H(i,v)$, the mean weight of the corresponding walk is
$(it_2 - (n-i)t_1)/n$, which is an increasing function of $i$.
Thus, once this exceeds the minimum value, $min$, found so far 
a critical term has been found and is already reflected in the value of $min$.
Thus, Algorithm~\ref{alg:MCM} returns the minimum cycle mean.

The minimum cycle mean is the ratio of edges of weight $-t_1$ to edges
of weight $t_2$ on a cycle of minimum mean.  This must also be true
for a critical walk, by Lemma~\ref{lem:meanBound}.  This ratio for the 
walks of interest in pass
$i$ is $(n-i)/i$, so the algorithm halts before the first pass
$i'$ such that $(n-i')/i' > \lambda$, and $i' = O(n/\lambda)$.  Thus,
Algorithm~\ref{alg:MCM} halts after $O(n/\lambda)$ passes.  

Using the approach of Section~\ref{sect:BF},
the operations in a pass take $O(m)$ time except for evaluating
$\max_{0 \leq j < i} (H(i,v) - H(j,v))/(n - l(j,v))$ for terms $H(i,v)$ 
of interest.  For any vertex $v$, $H(i,v)$ is
a term of interest for at most one value of $i$.  Therefore, the
cost of evaluating $\max_{0 \leq j < i} (H(i,v) - H(j,v))/(n - l(j,v))$ for 
terms of interests is bounded by the total number of dynamic programming 
table entries $H(j,w)$ for $0 \leq j \leq i$ and $w \in V$
computed by the algorithm, which is the number $n$ of them computed in each 
pass times $O(n/\lambda)$ passes.  This is $O(n^2/\lambda)$.
\end{proof}

\begin{algorithm}
\KwData{$G_d$, $(t_1,t_2)$}
\KwResult{The minimum cycle mean $\lambda$ of $G_d$}
$min := \infty$ \;
$H(0,v)$ := $0$ for all $v \in V$ \;
\For {$i$ := 1 to $\infty$} {
    \If {$min < (it_2 - (n-i)t_1)/n$} {return $min$}
    Compute $H(i,v)$ for all $v \in V$ from $H(i-1,v)$ for all $v \in V$ \;
    \For {each term $H(i,v)$ such that $l(i,v) = n$} {
        $k$ := $\max_{0 \leq j < i} (H(i,v) - H(j,v))/(n - l(j,v))$ \;
        \If {$k < min$} {$min = k$}
    }
}
\caption{Find the minimum cycle mean of $G_d$}\label{alg:MCM}
\end{algorithm}

\section{Appendix}

{\bf Theorem~\ref{thm:kextTime}:}
{\em Given a $k$-clique extendable ordering of a graph $G$, a maximum clique
can be found in $O(km^{k/2})$ time.}

\begin{proof}
Let $(v_1, v_2, \ldots, v_n)$ be the $k$-clique extendable ordering.
There are $O(m^{k/2})$ cliques in a connected graph with $m$ edges,
and they can be enumerated in $O(km^{k/2})$ time~\cite{ChibaNishizeki}.
If there are no $k$-cliques, a maximum clique can be found
in $O(km^{k/2})$ time by applying this algorithm to find
all $i$-cliques for $i < k$.

Otherwise, we denote each $k$-clique with a tuple 
$(u_1, u_2, \ldots, u_k)$ where $\{u_1, u_2, \ldots, u_k\}$ are
the elements of the clique and $u_1 < u_2 < u_3, \ldots, u_k$
in the $k$-clique extendable ordering.  We order the cliques lexicographically
by the reverse of its tuple (cliques sharing the last $j < k$ members
are consecutive in the list).  The lexicographic sort takes
$O(km^{k/2})$ time, since there are $O(m^{k/2})$ of them.
This list serves as the dynamic programming table, which has one
entry for each $k$-clique.
In addition we create a {\em block}, identified as $(u_2, u_3, \ldots, u_k)$ 
for each nonempty block of cliques that share $(u_2, u_3, \ldots, u_k)$ 
as their rightmost $k-1$ elements; they are consecutive
in the dynamic programming table.
This block is {\em relevant} to each $k$-clique of the form
$(u_2, u_3, \ldots, u_k,x)$.
We precompute
a pointer from each clique to its relevant block by 
lexicographically sorting the cliques by their first $k-1$ elements
in reverse order.  This order
is the order of their blocks in the table, so traversing this list and the
table concurrently allows assignment of the pointers from cliques to their
relevant blocks $O(km^{k/2})$ time.

The dynamic programming labels each $k$-clique with the size of the maximum 
clique of $G$ ending with its $k$ elements, and each block with the
maximum of the labels of cliques in the block.  This can be done
in lexicographic order:  when the last clique in
a block has been labeled, the block is labeled with the maximum of the
labels of the cliques in the block, and when clique
$(u_1, u_2, \ldots, u_k)$ is reached, its label is one plus the maximum
of the labels in its relevant block $(u_1, u_2, \ldots, u_{k-1})$, 
which is already labeled.  Traversing the
table performing these operations, using the precomputed pointers
to blocks, takes $O(m^{k/2})$ time.

The maximum label of any $k$-clique tells the size of a maximum
clique in the graph.  Let $K=(u_1, u_2, \ldots, u_k)$ be a $k$-clique
with maximum label.
To reconstruct the maximum clique of the graph, note that the last $k-1$ 
elements of this clique are $\{u_2, u_2, \ldots, u_k\}$.  Find the
remaining elements, as follows:  recursively find all but the last
$k-1$ elements of the largest clique ending on a $k$-clique in $K$'s relevant
block, which is empty if $K$ has no relevant block, and add $u_1$ to this 
result. 
\end{proof}

{\bf Theorem~\ref{thm:coloring}}:
{\em For $G$ in the class of dags where $\lfloor \lambda(G) \rfloor + 1 \leq k$,
there is a polynomial
$k$-approximation algorithm for the problem of finding a minimum
coloring of $G$.}

\begin{proof}
Given a nondegenerate dag $G$, find a satisfying assignment of utility values
for $(t_1, t_2)$ such that $t_2/t_1 = \lambda(G)$.

Let $x$ be the lowest
value assigned by the algorithm to any of the vertices, and let
$y$ be the highest.  
Partition the interval $[x,y]$ into {\em buckets} of the form
$[x,y] \cap [x + it_1, x+(i+1)t_1)$ for $i \geq 0$.  For each bucket, 
return the vertices whose $\alpha$
values lie in each bucket as the color classes.

Let $v$ be a vertex such that $\alpha(v) = x$.  In an optimum coloring, 
${\cal C}$,
removal of the color class containing $v$ removes a subset of
the set of vertices in the first $k$ buckets since it is an
independent set.
Removal of the color classes of the returned coloring that correspond
to the first $k$ buckets advances
the minimum $\alpha$ value among the remaining
vertices by $kt_1$. Removal of the of the color class in an optimum coloring
containing $v$ advances it by at most that much.
By induction on the number of vertices, we may assume that
the number of remaining color classes
of the returned coloring is at most $k$ times the number of color classes
in the remainder of ${\cal C}$.
Thus, for every color class in an optimum coloring, there
are at most $k$ in the coloring returned by the algorithm.
\end{proof}

{\bf Theorem~\ref{thm:cCover}:}  
{\em For $G$ in the class of dags where 
$\lfloor \lambda(G) \rfloor +1 = k$, there is a polynomial
$k$-approximation algorithm for the problem of finding a minimum
clique cover of $G$.}

\begin{proof}
Given a nondegenerate dag $G$, find a satisfying assignment of utility values
for $(t_1, t_2)$ such that $t_2/t_1 = \lambda(G)$.

Let $y$ be a value such that $[y, y+t_2]$ contains the $\alpha$ values of 
a maximum number of vertices.  Select $v$ from among the vertices
whose $\alpha$ value lie in $[y, y+t_2]$.  
Select $\{v = v_1, v_2, \ldots, v_j\}$ so that for $i > 1$,
$v_i$ minimizes $\alpha(v_i)$ over all vertices $x$ such that
$\alpha(x) > \alpha(v_{i-1}) + t_2$.  
Select $\{v = w_1, w_2, \ldots, v_{j'}\}$ such that for $i > 1$,
$w_i$ maximizes $\alpha(w_i)$ over all vertices $x$ such that
$\alpha(x) < \alpha(v_{i-1})-t_2$.
Let $K$ be the union of these two sets.
Because the pairwise differences
in $\alpha$ values are greater than $t_2$, $K$ is a clique.
Let this be one of the cliques in the clique cover.  Remove
it from the set of vertices, and recurse on the remaining vertices
to get the remaining cliques of the cover.

To see that this has an approximation ratio of at most $k$, 
let $X$ be the set of vertices whose $\alpha$ values are in $[y, y+t_2]$.  
Each clique of the clique cover returned by the algorithm removes one vertex
from $X$.  In a minimum clique cover, each pair of vertices must
have $\alpha$ values that differ by at least $t_1$.  Thus, no
clique can contain more than $k$ vertices from $X$.
The clique cover returned by the algorithm has at most $k$ times the
number of cliques as a minimum clique cover.
\end{proof}

{\bf Lemma~\ref{lem:meanBound}}:
{\em In $G_d$, the mean weight of a critical walk is less than or equal to the 
minimum cycle mean.}

\begin{proof}
Let $(t_1,t_2)$ be assigned arbitrarily.  Since we may apply the result
separately to each strongly connected component of $G_d$, we may assume that
$G_d$ is strongly connected.
Let $G'_d$ be the result of subtracting $c$ from every edge weight in $G_d$.
The mean weight of $C$, hence its total weight, is 0 in $G'_d$.
Since they all have the same length $n$, the paths of interest in $G_d$ are 
the same as they are in $G'_d$.

Out of all paths ending at a vertex on $C$, let $P$ be one of minimum 
weight in $G'_d$, let $w_1$ be its
weight in $G'_d$, and let $u \in C$
be the last vertex of $P$.  Let $W'$ be the walk of length $n - |P|$
obtained by walking round and round $C$, starting at $u$, let $v$
be the last vertex of $W'$, and let $W$ be the walk of length $n$
obtained by concatenating $P$ and $W'$.  Let $s$ be the first vertex
of $W$.

In $G'_d$, the weight $w$ of $W$ is equal to the minimum weight of a 
walk of any length ending at $v$, which is seen as follows.
Suppose there is a walk $W'$ of weight $w' < w$, ending
at $v$.  Let $w_2$ be the
weight of the portion of $C$ directed from $u$ to $v$, and let $w_3$
be the weight of the portion of $C$ directed from $v$ to $u$.  
Since $C$ has weight 0, $w = w_1 + w_2 > w'$.
Appending to $W'$ the portion of $C$ from $v$ to $u$
gives a walk ending at $u$ of 
weight $w' + w_3 < w_1 + w_2 + w_3$, and, since $C$ has weight 0,
this is just $w_1$.  Removing any cycles from this walk, we get
a path of weight $w_1$, contradicting that $P$
is a path of minimum weight to $u$.

Thus, $W$ is a walk of interest in $G'_d$, hence in $G_d$.
In $G'_d$, since there is a walk of length 0 and weight $0$
ending at $v$, so the weight of $W$ in $G'_d$, hence its mean weight in $G'_d$ 
is at most the minimum cycle mean of 0 in in $G'_d$.  Its mean weight 
in $G_d$ is at most the minimum cycle mean of $G_d$.

Since the edges of weight $-t_1$ are acyclic 
there is at least one edge of weight $t_2$ on $C$.
Since $W$ has length $n$ and $C$ is the only cycle on it,
$W$ makes at least one complete revolution of $C$, and
the weight of $W$ is $H(i,v)$ for some $i > 0$.  Therefore,
$l(i,v) = n$ and $H(i,v)$ is a term of interest, and the mean weight
of its walk of interest is the minimum cycle mean.
\end{proof}

{\bf Theorem~\ref{thm:karp2}:}
{\em In $G_d$, the minimum mean weight of a cycle is equal to \\
$\min_{\{(i,v)| l(i,v) = n\}} \max_{0 \leq j < i} (H(i,v) - H(j,v))/(n  - l(j,v))$}
\\

The proof requires a lemma:
\\

{\bf Lemma}:  If the minimum cycle mean is zero, then \\
$\min_{(i,v)| l(i,v) = n} \max_{0 \leq j < i} (H(i,v) - H(j,v))/(n  - l(j,v)) = 0$.

\begin{proof}
Suppose $l(i,v) = n$.
Because there are no negative cycles, there is a minimum-weight walk ending
at $v$ whose length is less than $n$.  Let its weight be $\pi(v)$.
$H(i,v) \geq \pi(v)$.  Also, $\pi(v) = \min_{0 \leq k < i} H(k,v)$, so
$H(i,v) - \pi(v) = \max_{0 \leq k < i} (H(i,v) - H(k,v)) \geq 0$,
and $\max_{0 \leq k < i} (H(i,v) - F(k,v))/(n-l(k,v)) \geq 0$.

Equality holds if and only if $H(i,v) = \pi(v)$.  We complete the
proof by showing that there exists $v$ such that there exists
$i$ where $l(i,v) = n$ and $H(i,v) = \pi(v)$.
Let $C$ be a cycle of weight zero and let $w$ be a vertex on $C$.
Let $P(w)$ be a path of weight $\pi(w)$ ending at $w$.  Then $P(w)$,
followed by any number of repetitions of $C$, is also a minimum-weight
walk to its endpoint.  After sufficiently many repetitions of $C$,
such an initial part of length $n$ will occur; let its endpoint be $w'$.
Then $l(i,w') = n$ for some $i$, and $H(i,w') = \pi(w')$.  Choosing
$v = w'$, the proof is complete.

{\bf Proof of Theorem~\ref{thm:karp2}}.  
Reducing the weight of each edge weight by a constant $c$ reduces the minimum
cycle mean by $c$, and $H(k,v)$ is reduced by $l(k,v)c$.
If $l(i,v) = n$, $(H(i,v) - H(k,v))/(n-l(k,v))$ is reduced by $c$,
and $\min_{l(i,v) = n} \max_{0 \leq k < i} (H(i,v) - H(k,v))/(n - l(k,v))$ is 
reduced by $c$.  The minimum cycle mean and
this expression are affected equally.
Choosing $c$ to be the minimum cycle mean and then applying the
lemma, we complete the proof.

\end{proof}

\end{document}